\documentclass[10pt]{amsart}

\usepackage{amssymb,amsmath,euscript}
\usepackage{mathtools}
\usepackage{breqn}
\usepackage{breqn}
\usepackage{tikz}
\usepackage{tikz}
\usepackage{tkz-berge}
\usepackage{enumitem}

\usepackage{}

\newtheorem{theorem}{Theorem}[section]
\newtheorem{lemma}[theorem]{Lemma}
\newtheorem{proposition}[theorem]{Proposition}

\newtheorem{corollary}[theorem]{Corollary}
\newtheoremstyle{case}{}{}{}{}{}{:}{ }{}
\theoremstyle{case}

\theoremstyle{definition}
\newtheorem{definition}[theorem]{Definition}
\newtheorem{example}[theorem]{Example}

\theoremstyle{remark}
\newtheorem{remark}[theorem]{Remark}
\def\Fq{{\mathbb F}_q}
\def\a{{\alpha}}

\def\FF{{\mathbb F}}

\def\Pml{{\mathbb P}^{{m\choose \ell}-1}}

\newcommand{\Ev}{\operatorname{Ev}}

\newcommand{\supp}{\mathrm{Supp}}

\def \Glm {G_{\ell, m}}
\def \glm {G_{2, m}}

\def \A {\mathbf{A}_\bullet}

\def \Ilm {\mathbb{I}(\ell, m)}
\def \Olm {\Omega_{\a}(\ell, m)}
\def \olm {\Omega_{\a}(2, m)}
\def \OA {\Omega(\A)}
\def \Liwpq {\mathcal L^{i}_{\mathcal W}(P, Q)}
\def \Kiwpq {\mathcal K^{i}_{\mathcal W}(P, Q)}
\def \C {C(\ell, m)}

\def \Calm {C_\a(\ell, m)}
\def \Cperp {C(\ell, m)^{\perp}}
\def \Calmp {C_\a(\ell, m)^{\perp}}
\def\calm {C_\a(2,m)}
\def \calmp {C_\a(2, m)^{\perp}}

\def \P {\overline{P}}

\def \Pclj {\overline{P}^{(j)}}

\def \Picl {\overline{P}^{(i)}}
\def \Piicl {\overline{P}^{(i-1)}}

\def \Ticl {\overline{T}^{(i)}}

\def \Ricl {\overline{R}^{(i)}}

\begin{document}
	
	\title[Decoding Certain Schubert Codes]{Majority Logic Decoding for Certain Schubert\\ Codes Using Lines in Schubert Varieties}
	\author{Prasant Singh}
\address{Department of Mathematics and Statistics,\newline \indent
	The Arctic Univerty of Norway(UiT),\newline \indent
	Hansine Hansens veg 18, 9019 Tromsø, Norway.}
\email{psinghprasant@gmail.com}

	\date{\today}
	
\begin{abstract}

In this article, we consider Schubert codes, linear codes associated to Schubert varieties, and discuss minimum weight codewords for dual Schubert codes. The notion of lines in Schubert varieties is looked closely at, and it has been proved that the supports of the minimum weight codewords of the dual Schubert codes lie on lines and any three points on a line in Schubert variety correspond to the support of some minimum weight parity check for the Schubert code. We use these lines in Schubert varieties to construct orthogonal parity checks for certain Schubert codes and use them for majority logic decoding. In some special cases, we can correct approximately up to $\lfloor (d-1)/2\rfloor$ many errors where $d$ is the minimum distance of the code.

\end{abstract}

\maketitle
\section{Introduction}
Let $q$ be a prime power and let $\Fq$ be a finite field with $q$ elements. Let $\ell$ and $m $ be two positive integers satisfying $1\le \ell\le m$ and let $\Glm$ be the  set of $\Fq$-rational points of the Grassmann variety of all $\ell$- planes in an $m$-dimensional space over the algebraically closed field $\overline{\mathbb{F}}_q$. It is well known that $\Glm$ can be identified with the set of all $\ell$ dimensional subspaces of $\mathbb{F}_q^m$ and it can be seen as an algebraic variety via the Pl\"ucker map. One can associate a linear code with every algebraic variety in a natural way \cite{TVN}. Therefore, it is natural to look at  codes associated to Grassmannians. Ryan \cite{Ryan, Ryan2} initiated the study of codes associated to Grassmannians over a binary field and Nogin \cite{Nogin} continued the study of these codes over any finite fields. The code associated to $\Glm$ is called a Grassmann code and is denoted by $\C$ and it has been proved that the Grassmann code $\C$ is an $[n, k, d]$ code where 
\begin{equation}
\label{eq: Grassmann}
n= {m\brack\ell},\quad k=\binom{m}{\ell},\quad\text{ and, }\quad d=q^{\ell(m-\ell)}.
\end{equation}
The Grassmann codes have been studied by several authors. For example, Nogin  determined the weight distribution of the code $C(2, m)$ in \cite{Nogin}, and that of $C(3, 6)$ in  \cite{Nogin1}. Furthermore, Kaipa, et.al \cite{KP} found the weight distribution of $C(3, 7)$, and the automorphism group of $\C$ was computed in \cite{GK}. Some of the generalized Hamming weights of these codes were studied in \cite{ GL, GPP, Nogin}. As for as  decoding of Grassmann codes is concerned, not much is known.  Recently \cite{BS} an attempt was made to find an error-correcting algorithm for $\C$. In this work, the majority logic decoder has been used and the lines in Grassmannians have been used to construct parity checks for the majority logic decoder. But unfortunately, we were able only to correct approximately up to $\lfloor(d-1)/2^{\ell+1}\rfloor$ many errors for $\C$ when $m$ is very large.\\

For every $\ell$-tuple $\a=(\a_1,\a_2\ldots,\a_\ell)$ satisfying $1\le\a_1<\a_2<\cdots<\a_\ell\le m$ there exists a subvariety of $\Glm$, known as the Schubert variety, and it is denoted by $\Olm$. Geometrically, The Schubert variety $\Olm$ in $\Glm$ is an algebraic subvariety given by certain Pl\"ucker coordinate hyperplanes. Since one can associate a linear code to every projective variety, it is natural to study the linear codes associated to Schubert varieties. The study of Schubert codes $\Calm$, codes associated to Schubert varieties $\Olm$, was initiated by Ghorpade-Lacahud in \cite{GL} and they conjectured that the minimum distance of $\Calm$ is $q^{\delta(\a)}$, where $\delta(\a)=\sum_{i=1}^\ell(\a_i-i)$. This conjecture is known as the MDC (the minimum distance conjecture). The conjecture was first proved for $\ell=2$ independently by Hana \cite{HC} and the Guerra-Vincenti \cite{GV}. The MDC was first proved, in generality, by Xiang \cite{X}. A different proof of the MDC was given, and an attempt to give a classification of the minimum weight codewords was made in \cite{GS}. In the case, $\ell=2$, the weight distribution of the Schubert code $\calm$ is known \cite{PS}. But not much is known about the decoding of Schubert codes. In \cite{BP}  it has been proved that the minimum weight codewords in $\C^\perp$ have their supports lying on lines in the Grassmannian $\Glm$. In this article, we give a classification of lines in Schubert varieties $\Olm$ and prove that the supports of the minimum weight codewords of the dual Schubert codes $\Calmp$ lie on  lines in $\Olm$ and if we choose any three points on a line in $\Olm$, there exists is a codeword of weight three in $\Calmp$. We also study the intersection of lines through a point on the boundary of discs in $\Glm$ centered at a point in $\Glm$. Further, we use lines in $\Olm$ to construct parity checks ``orthogonal" on each coordinates for the Schubert code $\Calm$ in the case, $\ell=2$ . Massey \cite{M} has used such  parity checks to perform majority logic decoding of a linear code. Therefore, one can use the set of parity checks obtained for $\calm$ in this article  and correct certain errors using   majority voting. Interestingly, in the case, $\a=(\a_1, m)$ with $\a_1=2, 3$ we are able to correct up to $\lfloor (d-1)/2\rfloor$ and $\lfloor (d-1)/2\rfloor-1$ many errors respectively, where $d$ is the minimum distances of the corresponding Schubert codes $\calm$.
\section{Preliminaries}
As in the introduction, let $q$ be a prime power and $\Fq$ be a finite field with $q$ elements. Let $\ell\le m$ be positive integers and let $V$ be an $m$ dimensional vector space over the field $\Fq$. The Grassmannian $G_\ell(V)$ of all $\ell$ planes in $V$ is defined  by 
$$
G_\ell(V):=\{ P\subset V:P\text{ is a linear subspace of dimension }\ell\}.
$$
It is easy to see that the cardinality $|G_\ell(V)|$ of the Grassmannian $G_\ell(V)$ is given by the Gaussian binomial coefficient ${m\brack \ell}_q$, where
\[
{m\brack\ell}_q:=\frac{(q^m-1)(q^{m}-q)\cdots (q^{m}-q^{\ell-1})}{(q^\ell-1)(q^{\ell}-q)\cdots (q^{\ell}-q^{\ell-1})}.
\]
The Pl\"ucker map embeds the Grassmannian $G_\ell(V)$  into the projective space $ \Pml$ as an algebraic variety. More precisely, let $\mathcal{B}$ be a fixed ordered basis of $V$. Define
\[
\mathbb{I}(\ell, m):=\{\a=(\a_1, \a_2,\ldots,\a_\ell)\in\mathbb{Z}^\ell:1\le \a_1<\a_2<\cdots<\a_\ell\le m)\}.
\]
Fix some linear order on $\Ilm$. Let $\Pml$ be the projective space of dimension ${m\choose \ell}-1$ over $\Fq$. For $P\in G_\ell(V)$, let $M_P$ denote the $\ell\times m$ matrix whose rows are coordinates of some basis of $P$ with respect to the basis $\mathcal{B}$. The Pl\"ucker map is defined by
\begin{equation}
\label{eq: pluckermap}
\mathrm{Pl}:G_\ell(V)\to \Pml \text{defined by }P\mapsto[P_\a]_{\a\in\Ilm}
\end{equation}
where $P_\a$ denote the $\ell\times \ell$ minor of $M_P$ corresponding to the columns of $M_P$ labeled by $\a$ and the $[P_\a]_{\a\in\Ilm}$ are taken in the same order as $\Ilm$. It is well known that the image of $G_\ell(V)$ is defined by the zero set of some irreducible quadratic polynomials, known as Pl\"ucker polynomials, and hence $G_\ell(V)$ is an algebraic variety. Further, the geometric structure of $G_\ell(V)$  depends only on $\ell$ and the dimension of $V$. To be precise, if $V$ and $V^\prime$ are two vector spaces of dimension $m$ over $\Fq$ then there exists an automorphism of $\Pml$ mapping $G_\ell(V)$ onto $G_\ell(V^\prime)$, and hence the varieties $G_\ell(V)$ and $G_\ell(V^\prime)$ are isomorphic. Therefore, we now write $\Glm$ to denote the Grassmann variety $G_\ell(V)$. Further, we think of $\Glm$ as a subset of a projective space $\Pml$ via the Pl\"ucker map. For more details on Grassmannians and Pl\"ucker polynomials, we refer to \cite{KL, Manivel}. We now define and describe lines in the Grassmannian $\Glm$. By a {\it line} in $\Glm$ we simply mean a line in $\Pml$ that is contained in $\Glm$. The lines in the Grassmannian can be parameterized by two subspaces $U$ and $W$ of $V$ of dimensions $\ell-1$ and $\ell+1$ respectively and satisfying $U\subset W$. Consider the following definition: 
\begin{definition}
	\label{def: lines}
Let $U \subset W$ be two subspaces of $V$ of dimensions $\ell-1$ and $\ell+1$ respectively. Define
\[
L(U, W):=\{P\in\Glm: U\subset P\subset W\}.
\]
\end{definition}

The sets $L(U, W)$ give a classification of lines on $\Glm$  \cite[Ch. 3.1]{MP}., i.e., every set $L(U, W)$ gives a line in $\Glm$ and every line in $\Glm$ is of the form $L(U,W)$ for some subspaces $U$ and $W$ as in the definition. The lines in $\Glm$ will play an important role in the decoding of Schubert codes. The Grassmannian has a natural notion of distance that is known as the injection distance. The injection distance between two points of Grassmannian is defined as \cite[Def. 2]{SK}: 

\begin{definition}
	Let $P,\;Q \in \Glm$ be given. The injection distance between $P$ and $Q$ is defined by $\mathrm{dist}(P,\;Q):=\ell-\dim(P \cap Q).$	
\end{definition}
Having the notion of distance in $\Glm$ we can talk about discs under injection distance of different radius around points of $\Glm$. 
\begin{definition}
	Let $P\in\Glm$ be a point and $0\le i\le \ell$ be a non-negative integer. The disc in $\Glm$ with center $P$  and radius $i$ is defined by
	\[
	\Picl:=\{Q\in\Glm: \mathrm{dist}(P,Q)\le i\}.
	\]
\end{definition}
Alternatively, the disc $\Picl$ can be defined by
\begin{eqnarray*}
	\Picl&=& \{Q\in\Glm: \dim (P \cap Q)\geq \ell-i\}\\
	&=& \{Q\in\Glm: \dim (P + Q)\leq \ell+i\}.
\end{eqnarray*}
Note that $\overline{P}^{(0)}=\{P\}$ and $\overline{P}^{(\ell)}=\Glm$. Further, $\overline{P}^{(1)}$ is the set of all points on $\Glm$ that lies on some line in $\Glm$ through $P$. For the sake of simplicity, we extend the definition of discs for every integer $i$ and set $\Picl=\emptyset$ for $i<0$ and $\Picl=\Glm$ for $i\geq\ell+1$. Note that these discs $\Picl$ are nested sets, i.e.
\[
\overline{P}^{(0)}\subset\overline{P}^{(1)}\subset\cdots\subset\overline{P}^{(\ell-1)}\subset \overline{P}^{(\ell)}.
\]
Geometrically, the set $\Picl\backslash\Piicl$ is the set of all points of $\Glm$ that lies exactly at distance $i$ with $P$. This gives a partition of the Grassmannian $\Glm$ as 
\begin{equation*}\label{eq:pic0}
\Glm= \bigsqcup\limits_{i=0}^{\ell}\left( \Picl \setminus \overline{P}^{(i-1)}\right).
\end{equation*}
A formula for the cardinality of the set $\Picl\backslash\Piicl$ is known \cite[Lemma 9.3.2]{BCN}. Later we will return to the discs in Grassmannians and study their intersections with lines through points from the boundaries of these discs. But for now, we will move to the main objects of this article, namely Schubert varieties and Schubert codes. First, we recall the definition of Schubert varieties in Grassmannians. 

Let $\a\in\Ilm$ be fixed. A partial flag $\mathbf{A}_\bullet: A_1\subset A_2\subset\cdots\subset A_\ell$ of subspaces of $V$ is said to be of dimension sequence $\a$ if $\dim A_i=\a_i$ for all $1\le i\le \ell$. Fix a partial flag  $\mathbf{A}_\bullet$ of dimension sequence $\a$. The Schubert variety of Grassmannian corresponding to the partial flag $\mathbf{A}_\bullet$ is defined and denoted by 
\[
\OA:=\{P\in\Glm:\;\dim(P\cap A_i)\geq i\text{ for }1\le i\le \ell\}.
\]
The restriction of the Pl\"ucker map to $\OA$ gives an embedding of $\OA$ into projective space. Further, $\OA$ is the zero set of all Pl\"ucker polynomials together with  certain coordinate hyperplanes and hence is an algebraic subvariety of $\Glm$. Note that, a priory, it seems that the definition of the Schubert variety $\OA$ depends on the partial flag $\mathbf{A}_\bullet$ but it is not true. More precisely, if $\mathbf{B}_\bullet$ is another partial flag of dimension sequence $\a$, then there exists an automorphism of the Grassmannian $\Glm$ that maps $\OA$ onto $\Omega(\mathbf{B}_\bullet)$. Therefore, we use the notation $\Olm$ to denote the  Schubert variety $\OA$ but the flag $\A$ is fixed. For a more detail study  of Schubert varieties, we again refer to\cite{KL, Manivel}.
\begin{example}
	\label{ex: schubert}
Let $P\in\Glm$ be fixed. For every $i$ satisfying $0\le i\le \ell$ the disc $\Picl$ is a Schubert variety $\Omega_\a(\ell, m)$, where $\a=(i+1,i+2,\ldots,\ell,m-i+1,m-i+2,\ldots, m)$ \cite[Lemma 3.2]{BS}.
\end{example}

Now we are ready to briefly define the Grassmann and Schubert codes. It is known \cite[Thm. 1.1.6]{TVN} that corresponding to every algebraic variety there is a naturally defined linear code (up to equivalence). The construction of the codes corresponding to Grassmann and Schubert varieties are as follows: Let $\mathbf{X}=(X_{ij})$ be an $\ell\times m$ matrix of indeterminates $X_{ij}$ over $\Fq$. For every $\a\in\Ilm$, let $X_a$ denotes the $\ell\times \ell$ minor of $\mathbf{X}$ corresponding to columns labeled by $\a$ and let $\Fq[X_\a]_{\a\in\Ilm}$ be the linear space spanned by all $X_\a$. Let $\Glm=\{P_1, P_2,\ldots, P_n\}$ be the set of  points on the Grassmannian in some fixed order, where $n={m\brack\ell}_q$ and let $M_{P_i}$ be the $\ell\times m$ matrix corresponding to $P_i$ as given in equation \eqref{eq: pluckermap}. Consider the evaluation map 
\begin{equation*}
\label{eq: defGrassCode}
\Ev: \Fq[X_\a]_{\a\in\Ilm}\to \Fq^n\text{ defined by }f(X_a)\mapsto c_f=(f(M_{P_1}),\ldots, f(M_{P_n})).
\end{equation*}
Since the Grassmannian is the zero set of Pl\"ucker polynomials and Pl\"ucker polynomials are quadratic irreducible polynomials, the evaluation map defined above is injective. The image of this map  is  called the Grassmann code and is denoted by $\C$. Therefore, for every codeword $c\in\C$, there exist a unique $f(X_a)\in \Fq[X_\a]_{\a\in\Ilm}$ such that $\Ev(F(X_\a))=c$ and the $i^{\it th}$ coordinate of $c$ is $f(M_{P_i})$ that we denote by $c_{P_i}$.   It is known \cite{Nogin, Ryan, Ryan2} that the Grassmann code $\C$ is an $[n, k, d]$ linear code where $n$, $k$, and $d$ are given in equation \eqref{eq: Grassmann}.

Schubert codes $\Calm$, the codes associated to the Schubert varieties $\Olm$, are codes obtained by puncturing the Grassmann code $\C$ on the complement of the Schubert variety $\Olm$ in the Grassmannian $\Glm$. In other words, there is a surjective projection map between Grassmann and Schubert codes defined by
\begin{equation}
\label{eq: defSchubcode}
\C\to \Calm\text{ defined by }\hat{c}=(c_{P_1},\cdots, c_{P_n})\mapsto c= (c_{P_i})_{P_i\in\Olm}
\end{equation}
The length and the dimension of the Schubert codes were determined in \cite{GT}. The minimum distance of $\Calm$ is known \cite{HC, GV} first in the case $\ell=2$, and then  \cite{GS, X} for general $\ell$. In particular, it has been proved that $\Calm$ is an $[n_\a, k_\a, d]$ code where
\begin{equation}
\label{eq: Schubparam}
n_{\a} = \sum_{\beta \le \a} q^{\delta(\beta)},\quad k_{\a} = \det_{1\le i,j\le \ell} \left( { {\a_j - j +1} \choose {i - j + 1}}\right) \text{ and }d=q^{\delta(\a)}.
\end{equation}
Here $\beta\le\a$ mean $\beta\in\Ilm$ and $\beta_i\le\a_i$ for $1\le i\le\ell$, and $\delta(\a)=\sum_{i=1}^\ell(\a_i-i)$. The Schubert variety corresponding to the dimension sequence $\a=(1, 2,\ldots, \ell)$ is a point and hence the corresponding Schubert code is a trivial code. Therefore, for this article we assume that $\a\in\Ilm$ is a general tuple but $\a\ne (1, 2,\ldots,\ell)$.
 Note that the coordinates of codewords in Grassmann codes (Schubert codes) are indexed by points of $\Glm$ (resp $\Olm$). Therefore, the support of any codewords of these codes (or its dual) can be thought of as a subset of the corresponding varieties. The next theorem from \cite[Thm. 24]{BP}  gives the geometric nature of the support of the minimum weight codewords of the dual Grassmann code.

\begin{theorem}
	\label{thm: BP}
	The minimum distance of the dual Grassmann code $\C^\perp$ is three. Further, the three points of $\Glm$ corresponding to the support of a minimum weight codeword of $\C^\perp$, lie on a line in the Grassmannian. Conversely, any three points on a line in $\Glm$, form the support of some minimum weight codeword in $\C^\perp$.
\end{theorem}
The minimum distance of the dual Schubert code is known \cite[Cor. 53]{Pinero} but the geometric structure of the support of these codes has not been studied like dual Grassmann codes. In the next section, we will prove that the supports of the minimum weight codewords of the dual Grassmann code $\Calm$ also lie on a line in $\Olm$.

\section{Lines in Schubert Varieties}
This section is the foundation for the decoding algorithm for some Schubert codes.
We will study the lines in the Schubert varieties $\Olm$ and their relation with the support of the parity checks of the Schubert code $\Calm$. We prove a result similar to Theorem \ref{thm: BP} for the Schubert codes. The lines in the Grassmannian $\Glm$ are of the form $L(U, W)$ \ref{def: lines}. We will see what conditions $U$ and $W$ have to satisfy to classify lines in $\Olm$. we will study the intersection of lines in Grassmannians $\Glm$ with discs $\Picl$ in $\Glm$. First, we describe the lines in Schubert varieties $\Olm$. By a {\it line} in $\Olm$ we mean a line of projective space that is contained in $\Olm$. Note that here we are treating $\Olm$ and its image under Pl\"uckler map as the same sets. Since the lines in Schubert varieties $\Olm$ are also  lines in the Grassmannian $\Glm$, they must be of the form $L(U, W)$ for some $U$ and $W$ as in Definition \ref{def: lines}. In the next proposition, we determine conditions on subspaces $U$ and $W$  that classify lines in $\Olm$. Recall that $A_1\subset A_2\subset\cdots\subset A_\ell$ is a partial flag of dimension sequence $\a$ and $\Olm$ is the corresponding Schubert variety.

\begin{proposition}
\label{prop: minimumwt}
Let $\mathcal G$ be a line in the projective space $\Pml$. Then $\mathcal G\subset\Olm$ if and only if $\mathcal G= L(U,W)$ for some subspaces $U$ and $W$ satisfying $U\subset W\subset A_\ell$ of dimensions $\ell-1$ and $\ell+1$ respectively,  $\dim (U\cap A_i)\ge i-1$ and $\dim(W\cap A_i)\ge i$ for every $1\le i\le \ell$ and $U, W$ satisfies any one of the following two conditions
\begin{enumerate}
	\item $\dim(U\cap A_i)\geq i$ for $1\le i\le \ell-1$.
	\item If $\dim (U\cap A_i)= i-1$ for some $1\le i\le \ell$, then $\dim(W\cap A_i)=i+1$.
\end{enumerate}
\end{proposition}
\begin{proof}
Let $\mathcal G\subset\Olm$ be a line in the Schubert variety.	Since $\Olm\subset\Glm$, $\mathcal G$ is a line $\Glm$ as well and hence there exist two subspaces $U\subset W$ of $V$ satisfying $\dim U=\ell-1$ and $\dim W=\ell+1$. If $T$ is a point on $L(U,W)$,  then $\dim(T\cap A_i)\geq i$ for $1\le i\le \ell$. Consequently, $\dim (U\cap A_i)\geq i-1$ and $\dim (W\cap A_i)\geq i$ for  $1\le i\le \ell$. Now if  item (1) of the Proposition is not true, then there exists some $1\le i\le \ell-1$ such that $\dim (U\cap A_i)=i-1$ for some $i$. Clearly, in this case, $\dim (W\cap A_i)\le i+1$. Let $Q_1,Q_2\in L(U, W)$ be two distinct points. Since $L(U, W)\subset \Olm,$  we have  $\dim(Q_j\cap A_i)\ge i$ for $j=1,2$. As $W=Q_1+Q_2$, $U=Q_1\cap Q_2$, and $(Q_1+Q_2)\cap A_i\supset Q_1\cap A-i + Q_2\cap A_i,$ we get
\begin{align*}
\dim (W\cap A_i) & =\dim ((Q_1+ Q_2)\cap A_i) \\
 & \ge\dim(Q_1\cap A_i +Q_2\cap A_i) \\
&\geq  i +i -\dim (U\cap A_i)\\
&= i + i -(i-1).
\end{align*}
This proves that $\dim(W\cap A_i)=i+1$.	 Conversely, let $L(U, W)$ be a line in $\Glm$   satisfying $U\subset W\subset A_\ell$, $\dim(U\cap A_i)\geq i-1$ and $\dim(W\cap A_i)\ge i$ for $1\le i\le \ell$. If $\dim(U\cap A_i)\geq i$ for every $i$,   or if $\dim(U\cap A_i)=i-1$ for some $i$, and $\dim(W\cap A_i)=i+1$ for some $i$, then  every $T\in L(U,W)$ satisfies $\dim(T\cap A_i)\geq \dim(U\cap A_i)\ge i$. Hence $L(U, W)\subset\Olm$. This completes the proof of the proposition.
\end{proof}	
 \begin{corollary}
 \label{cor: lineschubert}
 Let $L(U,W)$ be a line in $\Glm$. If $|L(U,W)\cap\Olm|\geq 2$ then $L(U,W)$ is a line $\Olm$.
 \end{corollary}
\begin{proof}
The proof is simple. Let $P,Q\in L(U,W)\cap\Olm$. As $U$ is a hyperplane in $P$ and $\dim(P\cap A_i)\geq  i$ for $1\le i\le \ell$, we get $\dim(U\cap A_i)\ge i-1$ and 	$\dim(W\cap A_i)\ge i$ for  $1\le i\le \ell$. Further, as $W=P+Q$,  and $P\subset A_\ell $, and $Q\subset A_\ell$, we get $W\subset A_\ell$.  Now let $\dim(U\cap A_i)=i-1$ for some $i$. Clearly, then $\dim(W\cap A_i)\le i+1$. Now as $W\cap A_i=(P + Q)\cap A_i\supset P\cap A_i +Q\cap A_i$, we get
\begin{align*}
\dim (W\cap A_i) & \geq \dim(P\cap A_i +Q\cap A_i) \\ 
&\geq  i +i -\dim (U\cap A_i)\\
&= i + i -(i-1)\\
&=i+1.
\end{align*}
This proves that $L(U,W)\subset\Olm$.
\end{proof}

The next theorem proves that the support of each minimum weight codeword of the dual Schubert codes $\Calmp$ lies on a line in $\Olm$, and vice-versa.

\begin{theorem}
	\label{thm: support}
	The minimum distance of $\Calmp$ is three. Further, the support of each minimum weight codewords of $\Calmp$ lie on a line in $\Olm$. Conversely, for any line $L(U, W)\subset\Olm$ and any three points $P,\; Q$ and $R$ on $L(U, W)$, there exists a codeword $\omega\in\Calmp$ such that 
	\[
	\supp(\omega)=\{P, Q, R\}
	\]
\end{theorem}

\begin{proof}
It is known  \cite[Cor. 53]{Pinero} that the minimum distance of the dual Schubert code $\Calmp$ is three. Recall that the Schubert code $\Calm$ is obtained by puncturing the Grassmann code $\C$ on the set $\Glm\setminus\Olm$.
To prove the theorem, we use equation \eqref{eq: defSchubcode} and Theorem \ref{thm: BP}. First, let $L(U, W)\subset\Olm$ be a line and $P,\;Q$ and $R$ be three points on $L(U,W)$. Then, by Theorem \ref{thm: BP}, there exists a minimum weight codeword $\hat{\omega}\in\Cperp$ such that $\supp(\hat{\omega})=\{P,Q,R\}$. Now let $\omega$ be a word that is obtained by puncturing $\hat{\omega}$ on all points in $\Glm\setminus\Olm$. We claim that, $\omega\in\Calmp$. Let $c\in\Calm$ be an arbitrary codeword  and let $\hat{c}\in\C$ be a codeword that is mapped to $c$ under the projection map in the equation \eqref{eq: defSchubcode}. Then
\[
\omega\cdot c=\sum_{P_i\in\Olm}\omega_{P_i}\cdot c_{P_i}=\sum_{P\in\Glm} \hat{\omega}_{P}\hat{c}_{P} =\hat{\omega}\cdot\hat{c}=0
\]
The second equality follows because of $\hat{\omega}_P=0$ for every $P\notin\Olm$. This proves that $\omega\in\Calmp$ and $\supp(\omega)=\{P, Q, R\}$. Consequently, every choices of three points on a line in $\Olm$ corresponds to the support of a minimum weight codeword of $\Calmp$. Conversely, let $\omega\in\Calmp$ be a weight three codeword. Let $\supp(\omega)=\{P,Q,R\}$ be the support of this codeword. We claim that these three points $P,\;Q$ and $R$ lie on a line in $\Olm$. Let $\hat{\omega}$ be a word which is an extension of the codeword $\omega$ as
\[
\hat{\omega}_P:=\begin{cases} \omega_P & \text{if } P\in\Olm,\\ 
0 & \text{if } P\in\Glm\setminus\Olm.
\end{cases}.
\]
Now, let $\hat{c}\in\C$ be an arbitrary point of the Grassmann code and let $c\in\Calm$ be the projection of $\hat{\omega}$. Note that
\[
\hat{\omega}\cdot\hat{c}=\sum_{P\in\Glm}\hat{\omega}_P\cdot\hat{c}_P=\sum_{P\in\Olm}{\omega}_P\cdot{c}_P=\omega\cdot c=0
\]
where the second equality follows from the fact that $\hat{\omega}_P=0$ for every $P\notin\Olm$ and the last equality follows from $\omega\in\Calmp$ and $c\in\Calmp$. This proves that $\hat{\omega}\in\Cperp$ and $\supp(\hat{\omega})=\{P, Q, R\}$. Now, from Theorem \ref{thm: BP}, it follows that these points $P,\;Q$ and $R$ lie on a line $L(U,W)$ in $\Glm$. Finally, as $\{P,Q,R\}\subset\Olm$ and $|L(U,W)\cap\Olm|\geq 3$, the result follows from Corollary \ref{cor: lineschubert}. This completes the proof of the theorem.
	
\end{proof}

\section{Disc and Lines in Grassmannians}
In this section, we study the intersection of certain lines with discs in $\Glm$ of different radius around some fixed points of $\Glm$. Since the support of a minimum weight parity checks of $\Calm$ lies on a line in the corresponding Schubert and hence Grassmann variety, the interaction of lines with these discs will be used extensively in the construction of orthogonal parity checks for certain Schubert codes. We use these parity checks for majority logic decoding for these codes. Throughout this section $P$ and $Q$ are two fixed points of the Grassmannian $\Glm$ satisfying  $\mathrm{dist(P,\; Q)}=i$, i.e., $Q\in\Picl\setminus\Piicl$. In the next theorems, we discuss the intersections  of lines passing through $Q$ with the disc of radius $i$ centered at $P$. 
\begin{theorem}
	\label{thm: line and closure}
	Let $P, Q\in \Glm$ be points with $Q\in\Picl\setminus\Piicl$ and let $L(U, W)$ be a line in $\Glm$ through $Q$. Then 
	\[
	|L(U, W)\cap \Picl|\geq 2 \iff  P\cap Q\subseteq U \text{ or }W\subseteq P+Q.
	\]
\end{theorem}
\begin{proof}
	Let $U, W$ and $Q$ be fixed such that $L(U, W)$ is a line through $Q$, i.e., let $U$ be an $\ell-1$-dimensional subspace of $Q$ and $W$ be an $\ell+1$-dimensional subspace of $V$ containing $Q$.
	
	Suppose $|L(U, W)\cap \Picl|\geq 2$. Let $T\in L(U, W)\cap \Picl$ be a point different from $Q$, i.e.,  $U=Q\cap T$ and $W= Q+ T$. Also, assume that $P\cap Q\nsubseteq U$. We claim that in this case $W\subseteq P + Q$. Since $P\cap Q\subseteq Q$, $\dim(P\cap Q)=\ell-i$ and $U$ is a hyperplane in $Q$ not containing $P\cap Q$, we obtain $\dim (P\cap Q\cap U)=\ell-i-1$. Since $P\cap Q\cap T= P\cap Q\cap U$ and $\dim (P\cap T)\geq \ell-i$, we can conclude
	\begin{align}\label{eq:dimineq} 
	\dim (Q+ (P\cap T)) & =\dim(Q) +\dim(P\cap T) - \dim(Q\cap P\cap T)\notag\\
	&\geq  \ell +(\ell - i)-(\ell - i - 1) \notag\\
	&= \ell +1.
	\end{align}	 
	Since $Q+ (P\cap T)\subseteq Q +T = W$, we see $W = Q + (P\cap T),$ by equation \eqref{eq:dimineq} and the fact that $\dim W=\ell+1$. Consequently, $W\subseteq Q + P$. This proves the first part of the theorem.\\
	Conversely, let $P\cap Q\subseteq U\subset Q\subset W$ and let $T\in L(U,W)$ be arbitrary. Then by definition $U\subset T$ and hence $P\cap Q\subset T$. Consequently $P\cap Q\subseteq P\cap T$. Hence $\dim (P\cap T)\geq\dim(P\cap Q)= \ell-i$, implying $T\in \Picl$. Finally, assume that $W\subseteq P + Q$ and let $ L(U, W)$ be a line through $Q$. Choose $T\in L(U, W)$ arbitrarily. Then $T\subset W\subseteq P + Q$ and hence $T + P\subseteq P + Q$. This proves that $\dim( T +P)\leq \ell + i$ and hence $T\in \Picl$.
\end{proof}	
Recall that $\overline{P}^{(1)}$ is the collection of all points in $\Glm$ that lies on some line through $P$.
\begin{corollary}
	\label{cor: i2case}
	Let $P,Q\in\Glm$ be two points and let $Q$ be on a line in $\Glm$ passing through $P$. Let $L(U,W)$ be a line through $Q$. Then $L(U, W)$ contains some point other than $Q$ that lies on a line through $P$ if and only if $U=P\cap Q$ or $W=P+Q$.
\end{corollary}
\begin{proof}
	The proof follows from the above theorem taking $i=1$.
\end{proof}	
\begin{corollary}
	\label{cor: line and closure}
	Let $Q\in \Glm$ be a point such that $Q\in\Picl\setminus\Piicl$ and let $L(U, W)$ be a line in $\Glm$ through $Q$. Then 
	\[
	|L(U, W)\cap \Picl|= 1 \iff  P\cap Q\nsubseteq U \text{ and }W\nsubseteq P+Q.
	\]
\end{corollary}

\begin{proof}
	Note that, $Q\in L(U, W)\cap\Picl$.	Therefore $	|L(U, W)\cap \Picl|\geq 1$. Now the result follows from the theorem. 
\end{proof}	

\begin{corollary}
	\label{cor: line contained in Picl}
	Let $P\in \Glm$ be a point and $L(U, W)$ be an arbitrary line in $\Glm$. Then for any $1\leq i\leq \ell$
	\[
	\text{either }|L(U, W)\cap \Picl|\leq 1 \text{ or }L(U, W)\subset \Picl
	\]
\end{corollary}

\begin{proof}
We can prove this using Theorem \ref{thm: line and closure}. But we give a shorter proof: It follows from Example \ref{ex: schubert} and Corollary \ref{cor: lineschubert}.
%
%
%
\end{proof}
In the previous theorem we studied the intersection of a line passing through a point $Q$ that is at distance $i$ from $P$, and the disc of radius $i$ around $P$. In the next theorem we study the intersection of such a line and a disc around $P$ of smaller radii.

\begin{theorem}
	\label{thm: line and closure2}
	Let $P$ and $Q$ be two points in $\Glm$ satisfying $Q\in\Picl\setminus\Piicl$ and let $L(U, W)$ be a line in $\Glm$ through $Q$. Then $|L(U, W)\cap \Pclj| =0$ for $j\leq i-2$ and $|L(U, W)\cap \Piicl|\leq 1$. Further,
	\[
	|L(U, W)\cap \Piicl|= 1 \iff  P\cap Q\subseteq U \text{ and }W\subseteq P+Q.
	\]
\end{theorem}

\begin{proof}
	Suppose $P, Q\in\Glm$ be given as in the lemma, and let $L(U, W)$ be any line through $Q$.  Since $\mathrm{dist(P,\;Q)}=i$ and $\mathrm{dist(T,\; Q)}\le 1$ for every $T\in L(U, W)$,  the triangle inequality gives $\mathrm{dist(P,\;T)}\ge i-1 $, and consequently 
\[
|L(U, W)\cap \Pclj |= 0 \text{ for all } j\leq i-2.
\]		
	Furthermore, since $Q\notin\Piicl$, and $Q\in L(U, W)$, clearly $L(U, W)\nsubseteq \Piicl$. From Corollary $\ref{cor: line contained in Picl}$ we get $|L(U, W)\cap \Piicl|\leq 1$.
	
	What is left to prove is the last statement. Assume $|L(U, W)\cap \Piicl|= 1$  and let $T\in L(U,W)\cap\Piicl$. Then $T\neq Q$, implying $U=T\cap Q$ and $W=T+Q$. Since $T\in\Piicl$ we get $\dim(P\cap T)\geq \ell-i+1$. Since we have already shown that $|L(U, W)\cap \overline{P}^{(i-2)}|=0$, we may conclude that $\dim(P\cap T)= \ell-i+1$. Since $U$ is a hyperplane in $T$, this implies $\dim (P\cap U)\geq \ell-i$. On the other hand, $P\cap U\subseteq P\cap Q$ and hence $\dim (P\cap U)\leq \ell-i$. Therefore, $\dim (P\cap U)= \ell-i= \dim (P\cap Q)$. This shows $P\cap U= P\cap Q$.
	This gives $P\cap Q\subseteq U$. On the other hand, since $T\subset W$ and $Q\subset W$ is a hyperplane, we get $\dim (P\cap W)=\ell-i+1$ and hence $\dim (P+ W)=\ell +i$. Note that 
	\[
	W+ P= (W + Q) + P= W+ (P + Q).
	\]
	This implies that  $P + Q \subseteq W+ P$. But since both of these subspaces are of dimension $\ell+i$, we get $W\subset P+Q$.\\
	Conversely, let $L(U, W)$ be a line through $Q$ satisfying $P\cap Q\subset U\subset W\subset P+ Q$. We already have seen that  $|L(U, W)\cap \Piicl|\leq 1$, it is sufficient to show that $L(U, W)\cap \Piicl$ is not empty. Observe that $W$ is a subspace of $P+Q$ of codimension $i-1$, as $\dim(P + Q)=\ell+i$ and $P$ is an $\ell$ dimensional subspace of $P+Q$. This implies that $\dim (P\cap W)\geq \ell-i+1$. But as $Q$ is a hyperplane of $W$, and $\dim (P\cap Q)=\ell-i$, we also may conclude that $\dim (P\cap W)\leq\ell-i+1$. Consequently, $\dim (P\cap W)=\ell-i+1$. Since $P\cap Q\subset U\subset Q$, we get $P\cap Q= P\cap U$. Hence $\dim (P\cap U)=\ell-i$, and as $U\subseteq W$, we get
	$$
	\dim (U+ (P\cap W))=\dim U + \dim (P\cap W) -\dim (P\cap U) =\ell.
	$$
	Note that $U\subset (U+ (P\cap W)\subset W $, hence $U+ (P\cap W) \in L(U, W)$. Further,  $\dim (P\cap (U+ (P\cap W))\geq \dim (P\cap W)=\ell-i+1$. Therefore, 
	\[
	U+ (P\cap W) \in L(U, W)\cap \Piicl.
	\]
	This completes the proof of the Theorem. 
	
\end{proof}
\begin{lemma}
	\label{lemma: closure of line}
	Let $P$, $T$ be two distinct points of $\Glm$ lying on a line $L(U, W)$ in $\Glm$. Let $Q\in\Glm$ be any point and $1\leq i\leq \ell$ be such that $Q\in\Picl\cap \Ticl$. Then $Q\in\Ricl$ for any point $R\in L(U, W)$.
\end{lemma}
\begin{proof}
	Since $P, T\in L(U, W)$, we have $U=P\cap T$ and $W= P+T$. Let $U\subset R\subset W$ be any point. Note that if $\dim (U\cap Q)\geq \ell-i$ then $\dim (R\cap Q)\geq \ell-i$ and hence $Q\in \Ricl$. So we assume $\dim (U\cap Q)< \ell-i$. Now since, $Q\in\Picl\cap \Ticl$ and $U$ is a hyperplane of $P$ and $T$, 
	\[
	\dim (Q\cap P) = \dim(Q\cap T)=\ell-i \text{ and }\dim (Q\cap U)=\ell-i-1.
	\]
	We also have $W=P + T$ and $(Q\cap P) + (Q\cap T)\subseteq Q\cap W$. Therefore we get,
	\begin{eqnarray*}
		\dim (Q\cap W)&\geq &\dim (Q\cap P)+ \dim (Q\cap T)-\dim (T\cap P\cap Q)\\
		&=& (\ell-i) + (\ell-i)- (\ell-i-1)\\
		&=& \ell -i +1.
	\end{eqnarray*}
	Now, since $R$ is a hyperplane of $W$, we get $\dim (Q\cap R)\geq\ell-i$ and consequently, $Q\in\Ricl$.
\end{proof}	
\begin{remark}
	\label{rmk: remark1}
	The results of this section are true for any Grassmann variety over any field. In other words, if we consider an arbitrary field $\FF$ and the Grassmannian $G_\ell(\FF^m)$ of all $\ell$-planes of $\FF^m$, then  Theorems \ref{thm: line and closure}, \ref{thm: line and closure2} and Lemma \ref{lemma: closure of line} proved in this section are true for Grassmannians $G_\ell(\FF^m)$  and lines in $G_\ell(\FF^m)$. But as we are interested in linear codes associated to  Grassmannians and Schubert varieties, we have defined these varieties over finite fields. 
\end{remark}
\section{Majority logic decoding for Schubert code $\calm$}
In this section, we use  majority logic decoding  for error correction for the Schubert code $\calm$. Therefore, for the rest of the article, we fix $\ell=2$, $\a=(\a_1, \a_2)$ and $A_1\subset A_2$ a partial flag of dimension sequence $\a$. From the definition of Schubert varieties, we know that $\olm$ is the collection of all two-dimensional subspaces of $A_2$ that intersect $A_1$ nontrivially. Therefore, without loss of generality, we may assume that $\a_2=m$ and $A_2=V$. Hence 
\[
\olm=\{P\in\glm: \dim(P\cap A_1)\ge 1\}.
\]
Let $\calm$ be the corresponding Schubert code. Then equation \eqref{eq: Schubparam} gives that $\calm$ is an $[n_\a, k_\a, d_\a]$ linear code, where
\begin{equation}
\label{eq: schubcode2}
	n_\a=\mathop{\sum_{\beta\in\mathbb{I}(2, m)}}_{\beta_1\le\a_1}q^{\beta_1+\beta_2-3},\quad k_\a=\frac{\a_1(2m-\a_1-1)}{2}, \text{ and }d_\a=q^{m+\a_1-3}.
\end{equation}
The idea is to use lines in $\olm$ to construct, for each point $P\in\olm$, sets of parity checks for $\calm$ ``orthogonal" on $P$ and use them for majority voting for errors. But before going into details we should recall the notion of  orthogonality for parity checks of a code. For a general reference on these topics, we refer to \cite[Ch 13.7]{MS} for the binary case and \cite[Ch 1]{M} for the $q$-ary case. As usual,  if $C$ is a code and $C^\perp$ is the dual of $C$ then a codeword of $C^\perp$ is called a {\emph parity check} of $C$
\begin{definition}
	Let $C$ be an $[n, k]$ code. A set $\mathcal J$ of $J$ parity checks of $C$ is said to be orthogonal on the $i^{\it th}$ coordinate if the $J\times n$ matrix $H$ having these $J$ parity checks as rows satisfies the following: 
	\begin{enumerate}
		\item Each entry in the $i^{\it th}$ column of $H$ is $1$.
		\item The Hamming weight of any other column of $H$ is at most $1$, i.e., if $j \neq i$ and the $j^{\it th}$ column of $H$ contains a non-zero entry in the $r^{\it th}$ row, then this is the only non-zero entry in this column.
	\end{enumerate}
\end{definition}
The following theorem from \cite{M} guarantees that if we can produce certain parity checks for a code orthogonal on each coordinate, then we can correct errors using majority logic.
\begin{theorem}\cite[Ch 1,Thm 1]{M}
	\label{thm: Maj decod}
	Let $C$ be an $[n,k]$ code such that for each $1 \le i \le n$, there exists a set $\mathcal J$ of $J$ orthogonal parity checks on the $i^{\it th}$ coordinate. Then the corresponding majority logic decoder corrects up to $\lfloor J/2 \rfloor$ errors. 
\end{theorem}
As we mentioned, we are going to use the lines from different points in $\olm$ to construct orthogonal parity checks, we should look into the lines in $\olm$ more closely. In the next two lemmas, we determine lines in $\olm$ passing through a fixed point $P\in\olm$.

\begin{lemma}
	\label{lemma: linesthroughpoint}
Let $P\in\olm$ be a point satisfying $P\subset A_1$. Then there exist
\[
{2\brack 1}_q{m-2\brack 1}_q
\]
many lines in $\olm$ passing through $P$.
\end{lemma}

\begin{proof}
In this case the counting is very simple. Take any one-dimensional subspace $U$ of $P$ and any three dimensional subspace of $V$ containing $P$. Then by Proposition \ref{prop: minimumwt} the corresponding line $L(U,W)$ passes through $P$ and is contained in $\olm$. Further, any line in $\olm$ through $P$ is given by a hyperplane $U\subset P$ and a three-dimensional space $P\subset W\subseteq A_2$. The number of such ordered pairs $(U, W)$, and hence the line in $\olm$ through $P$ is: 
\[
{2\brack 1}_q{m-2\brack 1}_q.
\]	
	
\end{proof}
\begin{lemma}
	\label{lemma: linesthroughpoint2}
	Let $P\in\olm$ be a point satisfying $P\nsubseteq A_1$. Then there exist
	\[
	q{\a_1-1\brack 1}_q + {m-2\brack 1}_q
	\]
 lines in $\olm$ passing through $P$.
\end{lemma}

\begin{proof}
	Clearly, any line in $\olm$ through $P$ is of the form $L(U,W)$ for some one-dimensional space $U$ and three-dimensional space $W$ satisfying $U\subset P\subset W$. From Proposition \ref{prop: minimumwt}, it is clear that  there are two possibilities for $U$, namely $\dim(U\cap A_1)=1, \text{ or }0$. Since $P\in\olm$ but $P\nsubseteq A_1$, we get $\dim (P\cap A_1)=1$ and hence if $L(U, W)$ is a line in $\olm$ through $P$ that satisfies $\dim(U\cap A_1)=1$ then $U=P\cap A_1$. Write  $U_0=P\cap A_1$. Then we first count lines $L(U, W)$ in $\olm $ through $P$ and $U=U_0$. In this case, any three-dimensional space $W\subset V$ satisfying $P\subset W$ gives a line $L(U_0, W)$ through $P$ and in $\olm$. This gives ${m-2\brack 1}_q$  lines of the form $L(U_0, W)$ in $\olm$ through $P$.

	Now let $L(U, W)$ be a line in $\olm$ through $P$ and let $\dim(U\cap A_1)=0$, i.e. $U\neq U_0$. In this case, we have $\dim (U\cap A_1)=0$ and hence by Proposition \ref{prop: minimumwt}, $P\subset W\subset V$  satisfies $\dim(W\cap A_1)=2$. Hence, $W= P+\langle x\rangle$ for some $x\in A_1$. In other words, $P\subset W\subset A_1 +\langle x\rangle$. 
This gives that the number of distinct such $W's$ are in one to one correspondence with one-dimensional subspaces of $(A_1+ P)/P$. Therefore, we get $({2\brack 1}_q-1)=q$ choices for $U\neq U_0$ and ${\a_1-1\brack 1}_q$  choices for $W$. Hence we get $({2\brack 1}_q-1){\a_1-1\brack 1}_q$  lines $L(U, W)$ in $\olm$ through $P$ with $U\neq U_0$. This proves that, if $P\nsubseteq A_1$, then there are 
\[
 q{\a_1-1\brack 1}_q + {m-2\brack 1}_q
\]
many lines in $\olm$ passing through $P$.
\end{proof}	
In the next theorem, we construct some orthogonal parity  checks for $\calm$ for each coordinate using the lines from the last two lemmas. Recall that for any $P\in\glm$ the disc $\overline{P}^{(1)}$ gives points in $\glm$ that lies on some line through $P$. For simplicity of the notation we write $\P$ to denote the set $\overline{P}^{(1)}$.
\begin{theorem}
	\label{thm: p1closure}
For every $P\in\olm$, there exists a set $\mathcal{J}_1(P)$ of parity  checks 	of $\calm$ of weight three satisfying
\begin{enumerate}
	\item For every $\omega\in \mathcal{J}_1(P)$ the support of $\omega$ contains $P$ and two other points from $\P\cap\olm$, i.e., the other two points lie on a line through $P$ in $\olm$.
	\item For $\omega_1,\;\omega_2\in\mathcal{J}_1(P)$, $\supp(\omega_1)\cap\supp(\omega_)=\{P\}$.
\end{enumerate}
Further, 
\begin{equation*}
\label{eq: p1closure}
|\mathcal{J}_1(P)|:=\begin{cases}\lfloor q/2\rfloor {2\brack 1}_q{m-2\brack 1}_q & \text{if } P\subset A_1,\\ 
\lfloor q/2\rfloor\left( q{\a_1-1\brack 1}_q + {m-2\brack 1}_q\right) & \text{if } P\nsubseteq A_1
\end{cases}.
\end{equation*}	
\end{theorem}

\begin{proof}
The proof is a simple consequence of Theorem \ref{thm: support} and Lemmas \ref{lemma: linesthroughpoint} and \ref{lemma: linesthroughpoint2}. In the last two lemmas, we have computed the number of lines through points of $\olm$. Now let $P\in\olm$ be an arbitrary point. Any line through $P$ contains $q$  points from $\olm$ other than $P$. Partition these $q$ points on a line other than $P$  into $q/2$ many  subsets of cardinality two (if $q$ is even) and  into $q/2$  subsets of cardinality two, one subset of cardinality one  (if $q$ is odd). Take all these subsets of cardinality two, together with $P$ each of them gives three points on a line in $\olm$. From Theorem \ref{thm: support}, every such set of three points is the support of some minimum weight codeword of $\calmp$. Hence we get $\lfloor q/2\rfloor$  such parity checks for each line through $P$ in $\olm$.  Let $\mathcal{J}_1(P)$ be the set of all parity checks obtained in this way. Item (1) is clearly satisfied. Now if $\omega_1,\;\omega_2\in\mathcal{J}_1(P)$, then their supports either lie on the same line or two different lines in $\olm$. If they are on the same line, by construction they contain $P$ and two other points from the subsets of the partition. In either case, these supports intersect only in $P$. This proves item $(2).$ Now the last part of the theorem follows from Lemmas \ref{lemma: linesthroughpoint} and \ref{lemma: linesthroughpoint2}.
\end{proof}	

Note that, if $\a_1=1$ then the corresponding Schubert variety $\olm$ is the projective space $\mathbb{P}^{m-2}$. Therefore, the corresponding Schubert code $\calm$ is the projective Reed-Muller code of order one and majority logic decoding for the first order projective Reed-Muller code is known \cite{BS}.  This is why we may assume that $\a_1\ge 2$. Let $P\in\olm$ be a point satisfying $P\subset A_1$
\begin{equation}
\label{eq: flagA}
\mathcal{U}_1\subset P=\mathcal W_2\subset\mathcal W_3\subset\cdots\subset \mathcal W_{\a_1}=A_1\subset\cdots\subset \mathcal W_m=V
\end{equation}
be a flag through $P$ satisfying $\dim \mathcal{U}_1=1$, $\dim\mathcal{W}_j=j$ for every $3\le j\le m$. For the rest of the article, whenever we consider lines through a point $P\in\olm$, we always mean a line in $\olm$.
\begin{lemma}
	\label{lemma: doubleclos}
 Let $P\in\olm$ be a point satisfying $P\subset A_1$ and let equation \eqref{eq: flagA} be a fixed flag through $P$. For $3\le i\le\a_1$, let $L(\mathcal U_1, \mathcal W)$ be a line through $P$ where $\mathcal W\subset\mathcal{W}_i$ and $\mathcal W\nsubseteq\mathcal{W}_{i-1}$. Let $Q\in L(\mathcal U_1, \mathcal W)$ be a point different than $P$ and let $L(U, W)$ be an arbitrary line through $Q$. Then $|L(U, W)\cap \overline{R}|=1$ for any $R$ on any line $L(\mathcal U_1, W_1)$ through $P$ satisfying $ W_1\subset \mathcal W_{i-1}$ iff $U\neq\mathcal U_1$ or $ W\nsubseteq \mathcal W_i$. The total number of such ordered pair $(U, W)$ of subspaces is:
	$$
	\left({2\brack 1}_q-1\right)\left({m-2\brack 1}_q-{i-2\brack 1}_q\right).
	$$
	
\end{lemma}
\begin{proof}
	Let $3\le i\le \a_1$ be an arbitrary integer and let $L(\mathcal U_1, \mathcal W)$ be a line through $P$ where $\mathcal W\subset\mathcal{W}_i$ and $\mathcal W\nsubseteq\mathcal{W}_{i-1}$. Let $L(\mathcal U_1, W_1)$ be an arbitrary line through $P$ for some $W_1\subseteq \mathcal W_{i-1}$ and let $R\in L(\mathcal U_1, W_1)$ be an arbitrary point. Since $\mathcal U_1\subset Q$,   $\mathcal U_1\subset R$ and $Q\neq R$ as $Q\nsubseteq W_{i-1}$, we get $\mathcal U_1= Q\cap R$. Consequently, $Q\in\overline{R} $ and hence $|L(U, W)\cap \overline{R}|\ge 1$ for any line $L(U, W)$ through $Q$. Let $L(U, W)$ be an arbitrary line through $Q$.  From Theorem \ref{thm: line and closure} we know that $|L(U, W)\cap \overline{R}|\ge 2$ iff $U=Q\cap R$ or $W =Q+ R$. Since $Q\subset W_i$ and $R\subset W_{i-1}$ we get $Q+ R\subset W_i$. Further, any $W\subset W_i$ containing $Q$ can be written as $Q+ R$ for some $R\subset W_{i-1}$. Since $R$ is an arbitrary point on an arbitrary line  $L(\mathcal U_1, W_1)$ through $P$ for some $W_1\subseteq W_{i-1}$, we get $|L(U, W)\cap \overline{R}|=1$ for any $R$ on any line $L(\mathcal U_1, W_1)$ for some $P\subset W_1\subset \mathcal W_{i-1}$ iff $U\neq\mathcal U_1$ or $ W\nsubseteq \mathcal W_i$. This completes the proof of the lemma. The last part of the lemma follows, as the number of choices for lines $L(U, W)$ through $Q$ is given by the number of choices for $U\subset Q$, $U\neq \mathcal U_1$, and $Q\subset W\subset V$ but $ W\nsubseteq \mathcal W_i$. But this number is: 
$$
\left({2\brack 1}_q-1\right)\left({m-2\brack 1}_q-{i-2\brack 1}_q\right).
$$
\end{proof}	
For every point $P\in \olm$ satisfying $P\subset A_1$ we fix a flag through $P$ as in equation \eqref{eq: flagA}. For a line $L(\mathcal U_1, \mathcal W)$  through $P$ where  $\mathcal W\subset\mathcal{W}_i$ but $\mathcal W\nsubseteq\mathcal{W}_{i-1}$ and for every $Q\in L(\mathcal U_1, \mathcal W)$ different from $P$, we define $\mathcal L^{i}_{\mathcal W}(P, Q)$ as the set of all lines $L(U, W)$ through $Q$, where $U\neq \mathcal U_1$ and $W\subset V$ but $ W\nsubseteq \mathcal W_i$. From the last part of the Lemma, we get that the cardinality of the set $\Liwpq$ is given by the formula in the Lemma. Note that $P+Q\subset \mathcal W_i$ therefore, from Corollary \ref{cor: i2case} we get $L(U, W)\cap \P=\{Q\}$ for every $L(U, W)\in\Liwpq$. In other words, every $L(U, W)\in\Liwpq$ has one point, namely, $Q\in \P$ and $q$ remaining points in $({\overline{P}^{(2)}}\setminus\P)\cap\olm$. Also, from the last lemma we have the following:
\begin{corollary}
	\label{cor: Imp1}
      Let $Q_1$ and $Q_2$ be two distinct points on the line $L(\mathcal U_1, \mathcal W)$  through $P$ where $\mathcal W\subset\mathcal{W}_i$ but $\mathcal W\nsubseteq\mathcal{W}_{i-1}$. If $L(U, W)\in\mathcal L^{i}_{\mathcal W}(P, Q_1)$	and   $L(U^\prime, W^\prime)\in \mathcal L^{i}_{\mathcal W}(P, Q_2)$ then $L(U, W)\cap L(U^\prime, W^\prime)=\emptyset$.
\end{corollary}
\begin{proof}
Let $L(U, W)$ and $L(U^\prime, W^\prime)$ be as given. Assume, if possible, that 	$L(U, W)\cap L(U^\prime, W^\prime)\ne\emptyset$ and $T\in L(U, W)\cap L(U^\prime, W^\prime) $. Since $P, \;Q_1$ and $Q_2$ are on a line and there is a line from $T$ to $Q_1$ and a line from $T$ to $Q_2$, from Lemma \ref{lemma: closure of line} we get that $P$ and $T$ are colinear. Consequently $T\in\P$. But this is a contradiction, as we have $L(U, W)\cap \P=\{Q_1\}$.
\end{proof}	
\begin{corollary}
	\label{cor: Imp2}
	Let $\mathcal W$ and $\mathcal W^\prime$ be two subspaces of $\mathcal W_i$ containing $P$ satisfying $\mathcal W\nsubseteq \mathcal W_{i-1}$ and $\mathcal W^\prime\nsubseteq \mathcal W_{i-1}$. Let $Q\in L(\mathcal U_1, \mathcal W)$ and $Q^\prime\in L(\mathcal U_1, \mathcal W^\prime)$ be points different from $P$. Then for $L(U, W)\in\Liwpq$ and $L(U^\prime, W^\prime)\in\mathcal L^{i}_{\mathcal W^\prime}(P, Q^\prime)$ we have $L(U, W)\cap L(U^\prime, W^\prime)=\emptyset$.
\end{corollary}
\begin{proof}
First, note that  $Q+ Q^\prime\subset \mathcal W_i$. As $\mathcal U_1=Q\cap Q^\prime$ we get $Q\in \overline{Q^\prime}$ hence $|L(U, W)\cap\overline{Q^\prime}|\ge 1 $. From Corollary \ref{cor: i2case} we know that $|L(U, W)\cap\overline{Q^\prime}|\ge 2 $ iff $U=\mathcal U_1$ or $W=Q +Q^\prime$. But since $L(U, W)\in\Liwpq$ we have $U\ne \mathcal U_1$ and $W\nsubseteq \mathcal W_i$. Hence,  $L(U, W)\cap\overline{Q^\prime}=\{Q\}$ and consequently $L(U, W)\cap L(U^\prime, W^\prime)=\emptyset$ as $L(U^\prime, W^\prime)\subset \overline{Q^\prime}$ and $Q\notin L(U^\prime, W^\prime)$.
	
\end{proof}	
\begin{corollary}
	\label{cor: Imp3}
	Let $3\le i,\; j\le \a_1$ be two distinct integers. Let $L(\mathcal U_1, \mathcal W)$ and $L(\mathcal U_1, \mathcal W^\prime)$ be two lines through $P$ satisfying $\mathcal W\subset \mathcal W_i$ but $\mathcal W\nsubseteq \mathcal W_{i-1}$ and $\mathcal W\subset \mathcal W_j$ but  $\mathcal W\nsubseteq \mathcal W_{j-1}$. Let $Q\in L(\mathcal U_1, \mathcal W)$ and $Q^\prime\in L(\mathcal U_1, \mathcal W^\prime)$ be points different from $P$. Then every line $L(U, W)\in\Liwpq$ and line $L(U^\prime, W^\prime)\in\mathcal L^{j}_{\mathcal W^\prime}(P, Q^\prime)$ intersects trivially.

\end{corollary}
\begin{proof}
	We may assume that $j<i$. Now the corollary follows from  Lemma \ref{lemma: doubleclos} as $Q^\prime$ lies on the line  $L(\mathcal U_1, \mathcal W^\prime)$ where $\mathcal W^\prime\subset \mathcal W_{i-1} $.
\end{proof}

 In the next theorem, we will use the parity checks obtained in the Theorem \ref{thm: p1closure}  and Lemma \ref{lemma: doubleclos} to construct parity checks of weight five such that the support of all these new parity checks contain $P$ and four other points from $({\overline{P}^{(2)}}\setminus\P)\cap\olm$. Further, the support of any such two parity checks shall have only $P$ in common. The idea of this construction is as follows:\\
Let $P\in\olm$ be a point satisfying $P\subset A_1$, and $\omega$ be a parity check in $\mathcal{J}_1(P)$ as described in Theorem \ref{thm: p1closure}. Let $\supp(\omega)=\{P,Q,R\}$. From the construction, we know that these three points lie on a line in $\olm$ through $P$. We consider lines in $\olm$ through $Q$ such that all points on these lines other than $Q$ lie in $(\overline{P}^{(2)}\setminus\P)\cap \olm$. We do the same for the point $R$. We will see that the numbers of such lines through $Q$ and $R$ are the same. We enumerate these lines, with the same index. Now consider subsets $\mathcal K_1(Q)\subset\mathcal{J}_1(Q)$ and $\mathcal K_1(R)\subset\mathcal{J}_1(R)$  such that the support of the parity checks in $\mathcal K_1(Q)$ contains $Q$ and two other points of $(\overline{P}^{(2)}\setminus\P)\cap \olm$ and  the support of the parity checks in $\mathcal K_1(Q)$ contains $Q$ and two other points of $(\overline{P}^{(2)}\setminus\P)\cap \olm$. From Lemma \ref{lemma: closure of line} we can see that the supports of a parity check $\omega_1\in\mathcal K_1(Q)$ and $\omega_2\in\mathcal K_1(R)$  are disjoint. If necessary, scale $\omega_i$ for $i=1, 2$ such that the parity check $\omega + \omega_1 +\omega_2$ does not contain $Q$ and $R$. We consider some other parity check from $\mathcal{J}_1(P)$ and repeat the process except this time the chosen lines (through $Q$ and $R$ ) must be chosen avoiding the support of previously constructed parity checks. The precise construction is given in the following theorem
\begin{theorem}
	\label{thm: p2closure}
	Let $P\in\olm$ be a point with $P\subseteq A_1$. Then there exists a set $\mathcal{A}_2(P)$ of parity checks of the Schubert code $\calm$ of weight five satisfying
	\begin{enumerate}
		\item The support of every $\omega\in \mathcal{A}_2(P)$  contains $P$ and four other points from $(\overline{P}^{(2)}\setminus\P)\cap \olm$.
		\item For $\omega_1,\;\omega_2\in\mathcal{A}_2(P)$, $\supp(\omega_1)\cap\supp(\omega_)=\{P\}$.
	\end{enumerate}
	Moreover, 
	\begin{equation*}
	\label{eq: p2closure}
	|\mathcal{A}_2(P)|=\frac{1}{q-1} (\lfloor q/2\rfloor)^2 \left(q^{m-2}{\a_1-1\brack 1}_q- \frac{q^{2\a_1-2}-1}{q^2-1}\right)
	\end{equation*}	
\end{theorem}
\begin{proof}
Let $P\subseteq A_1$ be a point in $\olm$ satisfying $P\subset A_1$. Fix a flag as in equation \eqref{eq: flagA} through $P$.
We prove by induction that for every $3\le i\le \a_1$, there exists a subset $\mathcal I_i(P)$ of parity checks for $\calm$ of weight five such that the support each $\omega\in \mathcal I_i(P)$ contains $P$ and four other points in the set $({\overline{P}^{(2)}}\setminus\P)\cap\olm$. Further,  every $\omega\in \mathcal I_i(P)$ and $\omega^\prime\in \mathcal I_3(P)\cup\cdots\cup \mathcal I_{i-1}(P)$ satisfies $\supp(\omega)\cap\supp(\omega^\prime)=\{P\}$, where $\mathcal I_{2}(P)=\emptyset$. Moreover,
$$
|\mathcal I_i(P)|=(\lfloor q/2 \rfloor)^2\left({i-2\brack 1}_q-{i-3\brack 1}_q\right)\left({2\brack 1}_q-1\right)\left({m-2\brack 1}_q-{i-2\brack 1}_q\right).
$$
These parity checks are obtained from the lines in $\Liwpq$ for every $P\subset \mathcal W\subset\mathcal W_{i} $ satisfying $\mathcal W\nsubseteq\mathcal W_{i} $ and $Q\in L(\mathcal U_1, \mathcal W)$ in Lemma \ref{lemma: doubleclos}. For every $3\le i\le \a_1$, let $L(\mathcal U_1, \mathcal W)$ be a line through $P$ satisfying $\mathcal W\subset \mathcal W_i$ but $\mathcal W\nsubseteq \mathcal W_{i-1}$. There are ${i-2\brack 1}_q-{i-3\brack 1}_q$ such lines. Each of these  lines $L(\mathcal U_1, \mathcal W)$ gives rise to $\lfloor q/2\rfloor$ such parity checks in the set $\mathcal{J}_1(P)$, such that the supports of these parity checks  lie on the line $L(\mathcal U_1, \mathcal W)$. Now let $\omega$ be such a parity check with let $Q_1\in\supp(\omega)$ and $Q_1\ne P$. Choose lines $L(U_1, W_1)\in\mathcal L^{i}_{\mathcal W}(P, Q_1)$, then we have $L(U_1, W_1)\cap \P=\{Q_1\}$. On the other hand, if $Q_2$ is any point on  $L(\mathcal U_1, \mathcal W)$ other than $P$ and $Q_1$, and $L(U_2, W_2)\in \mathcal L^{i}_{\mathcal W}(P, Q_2)$ then we have seen in the corollary \ref{cor: Imp1} that  $L(U_1, W_1)\cap L(U_2, W_2)=\emptyset$. Choose a parity check $\omega\in\mathcal J_1(P)$ whose support lies on the line $L(\mathcal U_1, \mathcal W)$ and points $Q_1,\; Q_2\in\supp(\omega)$ such that $Q_i\ne P$ for $i=1, 2$. Consider all lines in the sets $ \mathcal L^{i}_{\mathcal W}(P, Q_1)$  and $ \mathcal L^{i}_{\mathcal W}(P, Q_2)$  and parity checks in $\mathcal{J}_1(Q_1)$ and $\mathcal{J}_1(Q_2)$ obtained from these lines as in Theorem \ref{thm: p1closure}. There are $\lfloor q/2 \rfloor({2\brack 1}_q-1)({m-2\brack 1}_q-{i-2\brack 1})$  such parity checks. Enumerate them as $\omega_i(Q_j)$ for $j=1, 2$. Now scale each $\omega_i(Q_j)$, if necessary, such that $\omega + \omega_i(Q_1) + \omega_i(Q_2)$ does not contain $Q_1$ and $Q_2$ for any $i$ and hence is of weight five. On the other hand since all lines in $\mathcal L^{i}_{\mathcal W}(P, Q_1)$ and $\mathcal L^{i}_{\mathcal W}(P, Q_2)$ have only point $\{Q_1\}$ and $\{Q_2\}$, respectively, in common with $\P$, the remaining $q$ points in these lines are in $(\overline{P}^{(2)}\setminus\P)\cap \olm$. Therefore, the support of these parity checks contain $P$ and four other points in $(\overline{P}^{(2)}\setminus\P)\cap \olm$.  Now we can do it for each $\omega$ whose support lies on the line $ L(\mathcal U_1, \mathcal W)$, and in this way we get $\lfloor q/2\rfloor^2({2\brack 1}_q-1)({m-2\brack 1}_q-{i-2\brack 1}_q)$ many parity checks. Note that from  Corollary \ref{cor: Imp1} it follows that the supports of any two such codewords intersect only in $P$. We can argue like this for every $\omega\in\mathcal J_1(P)$ whose support lies on lines $L(\mathcal U_1, \mathcal W)$ for some $\mathcal W$ satisfying $\mathcal W\subset \mathcal W_i$ but $\mathcal W\nsubseteq \mathcal W_{i-1}$. There are $({i-2\brack 1}_q-{i-3\brack 1}_q)$ such lines and for each such line we use lines from $\Liwpq$ to construct weight five parity checks for the Schubert code $\calm$. We denote the set of these 
parity checks by $\mathcal I_i(P)$. Note that the supports of any two parity checks $\lambda,\;\lambda^\prime\in\mathcal I_i(P)$ intersect only in $P$. This simply follows from Corollaries \ref{cor: Imp1} and \ref{cor: Imp2}. Also, from  Corollary \ref{cor: Imp3}, it follows that if $\lambda\in \mathcal{I}_i(P)$ and $\lambda^\prime\in\mathcal{I}_{i-1}(P)\cup\cdots\cup \mathcal{I}_{3}(P)$, then the support of $\lambda$ and $\lambda^\prime$  intersect in $P$ only. Finally, we define  $\mathcal{A}_2(P)=\mathcal{I}_3(P)\cup\mathcal{I}_4(P)\cup\cdots\mathcal{I}_{\a_1}(P)$. Note that items $(1)$ and $(2)$ are satisfied for parity checks in $\mathcal{A}_2(P)$. Moreover, 
\begin{eqnarray*}
	|\mathcal{A}_2(P)|&=& \sum\limits_{i=3}^{\a_1}(\lfloor q/2 \rfloor)^2\left({i-2\brack 1}_q-{i-3\brack 1}_q\right)\left({2\brack 1}_q-1\right)\left({m-2\brack 1}_q-{i-2\brack 1}_q\right)\\
	&=& (\lfloor q/2 \rfloor)^2\left({2\brack 1}_q-1\right)\sum\limits_{i=1}^{\a_1-2}\left({i\brack 1}_q-{i-1\brack 1}_q\right)\left({m-2\brack 1}_q-{i\brack 1}_q\right)\\
	&=& (\lfloor q/2 \rfloor)^2\sum\limits_{i=1}^{\a_1-2} q^i(\frac{q^{m-2}-{q^i}}{q-1})\\
	&=&\frac{1}{q-1} (\lfloor q/2\rfloor)^2 \left(q^{m-2}{\a_1-1\brack 1}_q- \frac{q^{2\a_1-2}-1}{q^2-1}\right).
\end{eqnarray*}

This completes the proof of the theorem.

\end{proof}	
Next, we want to construct parity checks of weight five for the code $\calm$ that are orthogonal on coordinate $P$ for some $P\in\olm$ satisfying $P\nsubseteq A_1$. To do so we need a lemma similar to Lemma \ref{lemma: doubleclos} for points $P\in\olm$ and  $P\nsubseteq A_1$. Assume $P\in\olm$ is a point and $P\nsubseteq A_1$. Without loss of generality we may assume that $A_1=\langle e_1,\ldots, e_{\a_1}\rangle$ and $V=\langle e_1,\ldots, e_{\a_1},\ldots, e_m\rangle$ and $P=\langle e_1, e_m\rangle$. Define $\mathcal U_1=\langle e_m\rangle $ and for $2\le i\le m$ we define $\mathcal W_i=\langle e_1, e_m, e_2\ldots, e_{i-1}\rangle$. So we have the fixed flag through $P$
\begin{equation}
\label{eq: flagA2}
(0)\subset \mathcal U_1\subset P= \mathcal W_2\subset  \mathcal W_3\subset\cdots\subset  \mathcal W_{\a_1}\subset\mathcal W_{\a_1+1}\subset\cdots\subset V.
\end{equation}
\begin{lemma}
	\label{lemma: doubleclos2}
	For $3\le i\le\a_1+1$, let $L(\mathcal U_1, \mathcal W)$ be a line through $P$ where $\mathcal W\subset\mathcal{W}_i$ but $\mathcal W\nsubseteq\mathcal{W}_{i-1}$. Let $Q\in L(\mathcal U_1, \mathcal W)$ be a point different than $P$ and let $L(U, W)$ be an arbitrary line through $Q$. Then $|L(U, W)\cap \overline{R}|=1$ for any $R$ on any line $L(\mathcal U_1, W_1)$ for some $P\subset W_1\subset \mathcal W_{i-1}$ iff $U=Q\cap A_1$ and  $P\subset  W\subseteq  V$ but  $W\nsubseteq \mathcal \mathcal{W}_i$ or $U\ne Q\cap A_1\text{ and }\mathcal U_1$ and $P\subset W\subseteq \mathcal \mathcal{W}_{\a_1+1}$ but $W\nsubseteq \mathcal \mathcal{W}_i$. Further, the number of such lines through $Q$ is given by 
	$$
\left({m-2\brack 1}_q-{i-2\brack 1}_q\right) + \left(\left(q-1\right)\left({\a_1-1\brack 1}_q-{i-2\brack 1}_q\right)\right).
	$$
	
\end{lemma}

\begin{proof}
	The proof of the lemma is quite similar to Lemma \ref{lemma: doubleclos}.
	Let $3\le i\le\a_1+1$ and let $L(\mathcal U_1, \mathcal W)$ be a line through $P$ where $\mathcal W\subset\mathcal{W}_i$ but $\mathcal W\nsubseteq\mathcal{W}_{i-1}$. Let $Q\in L(\mathcal U_1, \mathcal W)$ be a point different from $P$. Note that for any such $Q$ we have $\dim(Q\cap A_1)=1$, i.e., $Q\in\olm$ but $Q\nsubseteq A_1$. From Lemma \ref{lemma: linesthroughpoint2} we know that there are two different kinds of lines through $Q_1$ namely lines of the form $L(U, W)$ where $U=Q\cap A_1$ and $Q\subset W\subset V$ or $U\neq Q\cap A_1$ and $Q\subset W\subset A_{\a_1+1}$. Let $L(\mathcal U_1, W_1)$ be an arbitrary line through $P$ for some $W_1\subset \mathcal W_{i-1}$ and let $R\in L(\mathcal U_1, W_1)$. Then $\mathcal U_1=Q\cap R$ and hence $Q\in\overline{R}$. Consequently, $|L(U, W)\cap\overline{R}|\geq 1$. From Theorem \ref{cor: i2case} we know that $|L(U, W)\cap\overline{R}|\geq 2$ iff $U=Q\cap R=\mathcal U_1$ or $W= Q+R$. Therefore, if we take $U\neq\mathcal U_1$ and $W\neq Q+R$ then $|L(U, W)\cap\overline{R}|=1$. Now if $U= Q\cap A_1$ we take any $Q\subset W\subset V$ but $W\nsubseteq \mathcal W_i$ or $U\ne Q\cap A_1,\; \mathcal U_1 $ and $Q\subset W\subset \mathcal W_{\a_1+1}$ but $W\nsubseteq \mathcal W_i$ then for any line $L(\mathcal U_1, W_1)$ satisfying $W_1\subset \mathcal W_{i-1}$ and any $R\in L(\mathcal U_1, W_1)$, we get $|L(U, W)\cap\overline{R}|= 1$. The converse is also true as any $W\subset \mathcal W_i$ containing $Q$ can be written as $Q+R$ for some $R\subset \mathcal W_{i-1}$. Finally, the number of lines $L(Q\cap A_1, W)$ through $Q$ satisfying $Q\subset W\subset V$ but $W\nsubseteq \mathcal W_i$ is  $({m-2\brack 1}_q -{i-2\brack 1}_q)$  and the number of lines $L(U, W)$ through $Q$ satisfying $U\ne Q\cap A_1,\; \mathcal U_1 $ and $Q\subset W\subset \mathcal W_{\a_1+1}$ but $W\nsubseteq \mathcal W_i$ is $(q-1)({\a_1-1\brack 1}_q-{i\brack 1}_q)$. Therefore, the total number of such lines through $Q$ is:  
	$$
\left({m-2\brack 1}_q-{i-2\brack 1}_q\right) + \left(\left(q-1\right)\left({\a_1-1\brack 1}_q-{i-2\brack 1}_q\right)\right).
$$
\end{proof}	
For $P\in \olm$ and $P\nsubseteq A_1$ let a flag through $P$ be fixed as in equation \eqref{eq: flagA2}. For every line $L(\mathcal U_1, \mathcal W)$  through $P$ satisfying $\mathcal W\subset\mathcal{W}_i$ but $\mathcal W\nsubseteq\mathcal{W}_{i-1}$ and for every $Q\in L(\mathcal U_1, \mathcal W)$ different from $P$ we denote by $\mathcal K^{i}_{\mathcal W}(P, Q)$ the set of lines $L(U, W)$ through $Q$ as obtained in the Lemma \ref{lemma: doubleclos2}. From the last part of the Lemma, we get that the cardinality of the set $\Kiwpq$ are given by the formula in the Lemma. Now all the properties discussed in Corollaries \ref{cor: Imp1}, \ref{cor: Imp2} and \ref{cor: Imp3} satisfied by the set $\Liwpq$ are also satisfied by these $\Kiwpq$.  In the next theorem, we determine some weight five parity checks for $\calm$, which are orthogonal on $P$ for $P\in\olm$ and $P\nsubseteq A_1$.
\begin{theorem}
	\label{thm: p2closurepart2}
	For every $P\in\olm$ with $P\nsubseteq A_1$, there exists a set $\mathcal{B}_2(P)$ of parity checks of $\calm$ of weight five satisfying:
	\begin{enumerate}
		\item The support of every $\omega\in \mathcal{B}_2(P)$  contains $P$ and four other points from $(\overline{P}^{(2)}\setminus\P)\cap \olm$.
		\item For $\omega_1,\;\omega_2\in\mathcal{B}_2(P)$, $\supp(\omega_1)\cap\supp(\omega_)=\{P\}$.
	\end{enumerate}
	Moreover, 
	\begin{equation*}
	\label{eq: p2closurepart2}
	|\mathcal{B}_2(P)|=(\lfloor q/2\rfloor)^2 \left((q^{m-2} +q^{\a_1-1}){\a_1-1\brack 1}_q- q^2\frac{q^{2\a_1-4}-1}{q^2-1}\right)
	\end{equation*}	
\end{theorem}
\begin{proof}
The proof of this theorem is exactly the same as the proof of the Theorem \ref{thm: p1closure}. For $P\in \olm$ and $P\nsubseteq A_1$ we fix a flag through $P$ as in equation \eqref{eq: flagA2}. For every $3\le i\le \a_1+1$ and for lines  $L(\mathcal U_1, \mathcal W)$  through $P$ satisfying $\mathcal W\subset\mathcal{W}_i$ and $\mathcal W\nsubseteq\mathcal{W}_{i-1}$ we construct parity checks as in Theorem \ref{thm: p2closure} using lines in $\Kiwpq$ for $Q\in L(\mathcal U_1, \mathcal W)$ different from $P$. Therefore, for $3\le i\le \a_1+1$ we get a set $\mathcal I_{i}(P)$ of parity checks for the Schubert code $\calm$ such that for any  $\lambda\in \mathcal I_{i}(P) $ we get $\supp(\lambda)\cap \P=\{P\}$ and the supports of any two different parity checks  $\lambda,\;\lambda^\prime\in\mathcal I_i(P)$ intersect only in $P$. Like in Theorem  \ref{thm: p2closure} we get

\begin{align*}
|\mathcal I_i(P)| = &\lfloor q/2\rfloor^2\left({i-2\brack 1}_q-{i-3\brack 1}_q\right)\left({m-2\brack 1}_q-{i-2\brack 1}_q\right)\\ 
& + \lfloor q/2\rfloor^2\left({i-2\brack 1}_q-{i-3\brack 1}_q\right)(q-1)\left( {\a_1-1\brack 1}_q-{i-2\brack 1}_q\right).
\end{align*}
Further for $3\le j\le \a_1+1$ and $j\ne i$ if $\lambda\in \mathcal I_{i}(P) $ and $\lambda^\prime\in \mathcal I_{j}(P) $ we have $\supp(\lambda)\cap \supp(\lambda^\prime)=\{P\}$. Now we define $\mathcal{B}_2(P)= \mathcal I_{3}(P)\cup\cdots \mathcal I_{\a_1+1}(P)$. Clearly, parity checks of $\mathcal{B}_2(P)$ satisfy items $(1)$ and $(2)$ of the theorem. Further, 
\begin{eqnarray*}
	|\mathcal{B}_2(P)|&=& \sum\limits_{i=3}^{\a_1+1}\lfloor q/2\rfloor^2\left({i-2\brack 1}_q-{i-3\brack 1}_q\right)\left(\left({m-2\brack 1}_q-{i-2\brack 1}_q\right) + (q^{\a_1-1}-q^{i-2})\right)\\
	&=& \lfloor q/2\rfloor^2\sum\limits_{i=1}^{\a_1-1} q^{i-1}\left(\frac{q^{m-2}-q^i}{(q-1)}\right) + \lfloor q/2\rfloor^2\sum\limits_{i=1}^{\a_1-1}q^{i-1} \left(q^{\a_1-1}- q^i\right) \\
	&=& \frac{\lfloor q/2 \rfloor)^2}{(q-1)}\left((q^{m-2}+(q-1)q^{\a_1-1}){\a_1-1\brack 1}- q^2\frac{q^{2\a_1-2}-1}{q^2-1}\right)\\
	\end{eqnarray*} 
\end{proof}

Now combining Theorems \ref{thm: p1closure}, \ref{thm: p2closure}, and \ref{thm: p2closurepart2}, we get the following.

\begin{theorem}
	\label{thm: majlogdecod}
Let $2\le \a_1\le m-1$ be positive integers, and let $\calm$ be the corresponding Schubert code.	Using the majority logic decoding we can correct up to $\lfloor J/2\rfloor$ errors for the Schubert code $\calm$, where
	\begin{multline*}
	J = \frac{\lfloor q/2 \rfloor)^2}{(q-1)}\left((q^{m-2}+(q-1)q^{\a_1-1}){\a_1-1\brack 1}- q^2\frac{q^{2\a_1-2}-1}{q^2-1}\right)\\
	  + \lfloor q/2\rfloor\left( q{\a_1-1\brack 1}_q + {m-2\brack 1}_q\right)
	\end{multline*}
\end{theorem}
\begin{proof}
Let $P\in\olm$ be an arbitrary point. If $P\subset A_1$, then we consider the set of parity checks of $\calm$ obtained from Theorem \ref{thm: p1closure} and  \ref{thm: p2closure} for this point $P$ and form the set $\mathcal J(P)=\mathcal J_1(P)\cup\mathcal{A}_2(P)$. Since the parity checks of $\mathcal J_1(P)$ and $\mathcal{A}_2(P)$ are orthogonal on the coordinate $P$, and as any point other than $P$ from the support of any parity checks in $\mathcal J_1(P)$ lies in $\P\cap\olm$, while the support of any parity check in $\mathcal J_1(P)$ lies in $(\overline{P}^{(2)}\setminus\P)\cap\olm$, the parity checks of  $\mathcal J(P)$ are orthogonal on $P$. Similarly, if $P\nsubseteq A_1$ we can repeat the argument for the set $\mathcal J(P)=\mathcal J_1(P)\cup\mathcal{B}_2(P)$. Now note that the cardinality of the set of parity checks orthogonal on $P$ is smaller in the case when $P\nsubseteq A_1$ and in this case the cardinality is exactly the $J$ given in the theorem. This proves that for every $P\in\olm$ there are at least $J$ parity checks orthogonal in $P$. Hence, using  Theorem \ref{thm: Maj decod}, we get that using majority logic decoding we can correct up to $\lfloor J/2\rfloor$ many errors for the Schubert code $\calm$. This completes the proof of the theorem.
\end{proof}	

\begin{remark}
	\label{rmk: remark2}
\begin{enumerate}\item When $\a_1=1$ the corresponding Schubert code is isomorphic to the first order Projective Reed--Muller code of length $(q^{m-1}-1)/(q-1)$ and minimum distance $d=q^{m-2}$.  In this case we can use the parity checks obtained in the second part of the Theorem \ref{thm: p1closure}, i.e, in the case, when $P\nsubseteq A_1$ to perform the majority logic decoding. In fact, over the binary field, we can correct up to $\lfloor (d-1)/2\rfloor$ errors \cite{BS}.
	\item  If we calculate the value of $J$ over the binary field $\mathbb{F}_2$,  we get 
	\[
	J= 2^{m +\a_1-3} + 2^{\a_1-1}- \frac{2^{2\a_1-2}-4}{3}-3.
	\]
	We know from equation \eqref{eq: schubcode2} that the minimum distance of the Schubert code $\calm$ is $d=2^{m+\a_1-3}$ and therefore one would like to be able to correct up to $\lfloor(d-1)/2\rfloor$ errors which, in this case, is $2^{m-\a_1-4}-1$. On the other hand, in this case we have
    \[
        \lfloor J/2\rfloor= 2^{m+\a_1-4} + 2^{\a_1-2}- \frac{2^{2\a_1-3}-2}{3}-2.
     \]	
     Therefore it appears that the smaller the $\a_1$,  the better the error correction.
     \item In the case  $\a_1=2$, the Schubert code $\calm$ has minimum distance $d= q^{m-1}$. If $q$ is even then we get
   \[
   J= \frac{(q+2)q^{m-1} + q^3 -2q^2- 2q}{4(q-1)}.
   \] 
   Note that, if $q>4$ and even, we can write $J$ as 
   $$
   J> q^{m-1}/4 + 3q^{m-1}/4(q-1) +q(q-1)/4 -1
   $$
   Hence, it appears that in the even case $q=2^r$ and $r\to\infty$, using majority logic decoding, we can more than   $\lfloor d/8\rfloor + \lfloor 3q^{m-2}/4\rfloor$  errors for the code $\calm$.
\end{enumerate}
\end{remark}
\begin{corollary}
	\label{cor: binarycase}
	Using the majority logic decoding for the binary Schubert code $\calm$ we can correct
	\begin{enumerate}
		\item Up to $\lfloor(d-1)/2\rfloor$ many errors when $\a_1=2$.
		\item Up to $\lfloor(d-1)/2\rfloor-1$ many errors when $\a_1=3$.
		\item Up to $\lfloor(d-1)/2\rfloor-7$ many errors when $\a_1=4$.
	\end{enumerate}
\end{corollary}
\begin{proof}
The proof follows from the item (2) of  Remark \ref{rmk: remark2} and inserting the values of $\a_1$ in the formula for $\lfloor J/2\rfloor$ 
\end{proof}

\section{Acknowledgements}

The author would like to express his gratitude to the Indo-Norwegian project supported by Research Council of Norway (Project number 280731), and the DST of Govt. of India. Thanks are also due to Prof. Trygve Johnsen for his remarks on the initial  drafts of the article.

\end{document}